\documentclass[aps,prx,superscriptaddress,amsfonts,amsmath,amssymb,showpacs,floatfix,reprint,longbibliography]{revtex4-2}
\RequirePackage[l2tabu,orthodox]{nag}
\usepackage{sidmid}
\hypersetup{
pdfsubject = {},
pdftitle = {},
pdfauthor = {}
}
\makeatletter
\AddToHook{begindocument}{%
  \@ifpackageloaded{babel}{}{%
    \@ifpackageloaded{polyglossia}{}{%
      \renewcommand{\selectlanguage}[1]{}%
    }%
  }%
}
\makeatother
\begin{document}

\title{Local autonomous inference machines for quantum LDPC codes}
\author{Siddhant Midha}
\email{siddhantm@princeton.edu}
\affiliation{Princeton Quantum Initiative, Princeton University, Princeton, NJ 08544}

\author{Dmitry A. Abanin}
\affiliation{Princeton Quantum Initiative, Princeton University, Princeton, NJ 08544}
\affiliation{Department of Physics, Princeton University, Princeton, NJ 08544}
\affiliation{\'{E}cole Polytechnique F\'{e}d\'{e}rale de Lausanne (EPFL), 1015 Lausanne, Switzerland}

\begin{abstract}
We introduce a local and distributed framework for decoding quantum LDPC codes. Built atop belief propagation (BP), the architecture combines strictly local processing of syndrome information, local feedback, autonomous operation, and highly parallelized decoding. We establish the framework at two complementary levels. First, we show that any code exhibiting a threshold under iterative BP decoding can be promoted to an autonomous local measurement-and-feedback dynamics while preserving that threshold. We then turn to codes for which naive BP does not itself exhibit threshold behavior. Our key observation is that BP can nevertheless provide sufficiently accurate \emph{local} marginal information: active syndrome defects use these local beliefs to \emph{move}, pair, and annihilate through asynchronous feedback. We demonstrate numerically that this construction exhibits threshold behavior in the point-like sectors of the two- and three-dimensional toric codes, where conventional BP decoding fails. Exploiting the graph-generality of BP, we then study the inference dynamics on the membrane-like sector of the three-dimensional toric code and in a family of bivariate-bicycle quantum LDPC codes.

\end{abstract}
\maketitle

\begin{figure}
    \centering
    \includegraphics[width=1.0\linewidth]{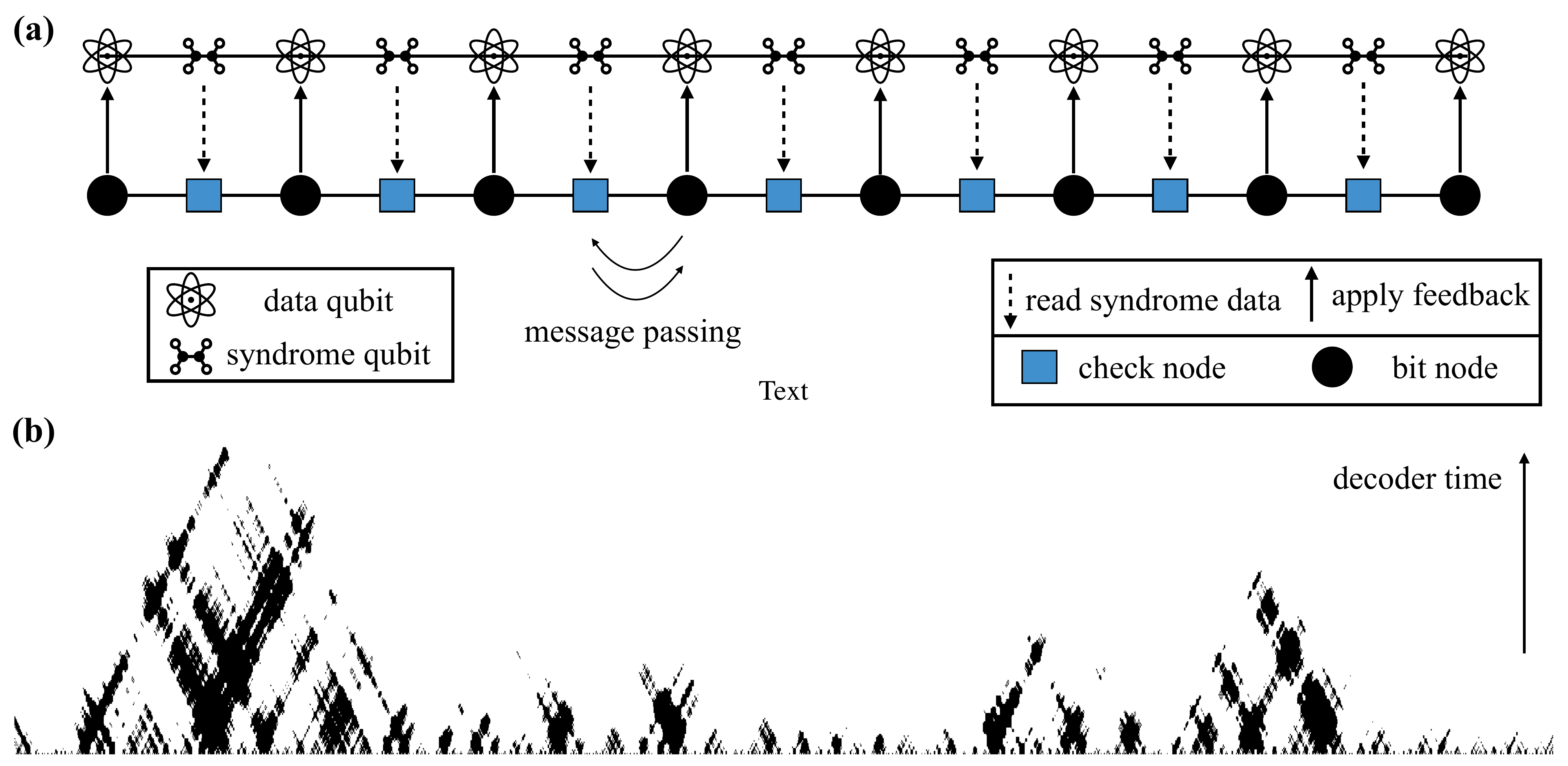}
    \caption{\textbf{Local autonomous inference}. (a) We consider the quantum processor, the \emph{device} consisting of data qubits undergoing $Z-$errors and syndrome qubits performing perfect measurements laid out according to the $X-$parity check matrix $H_X$. A series of classical processors performing message passing is laid out with the same geometry as the device. The device provides local information to the processors with syndrome data, and feedback is applied back onto the device. (b) Example trajectory of the residual error on the device as a function of decoder time (going upwards) on the repetition code with $d=2000$ at physical error rate $p=0.45$.}
    \label{fig:fig_setup}
\end{figure}

\textit{Introduction.---}Quantum error correction is essential for scalable, fault-tolerant quantum computation~\cite{TerhalQECReview,gottesmanIntroductionQuantumError2009d}. A key challenge is \emph{decoding}: the task of inferring an appropriate correction from the measured syndrome. As quantum processors scale, this inference task must itself remain efficient, with communication and processing requirements compatible with the locality and timescales of the underlying hardware~\cite{fowlerPracticalClassicalProcessing2012,ioliusDecodingAlgorithmsSurface2024b}. At the same time, local measurement-and-feedback dynamics closely related to ``active'' error correction have emerged as a setting for realizing and studying nonequilibrium quantum phases of matter~\cite{ODeaAbsorbing,FriedmanFeedback}. These two perspectives therefore motivate decoding architectures based on strictly local information processing and feedback.

The decoding problem for quantum codes has been extensively studied, yielding a wide range of efficient and high-performance algorithms. Many widely used decoders, however, require processing syndrome information over length scales that grow with the system size. For ``matchable'' topological codes, minimum-weight perfect matching (MWPM) formulates decoding as a global matching problem among syndrome defects \cite{dennisTopologicalQuantumMemory2002,fowlerPracticalClassicalProcessing2012,fowlerMinimumWeightPerfect2014}. Clustering-based decoders, such as union-find, grow and merge syndrome clusters over progressively increasing length scales \cite{delfosseUnionFindDecoderQuantum2021,delfosseAlmostlinearTimeDecoding2021a,chanActisStrictlyLocal2023}, while modern matching implementations efficiently exploit the sparsity and locality of the decoding graph~\cite{petersonDistributedBlossomAlgorithm2022,wuFusionBlossomFast2023,higgottSparseBlossomCorrecting2025}. Renormalization-group based decoders similarly exploit a hierarchy of increasing length scales~\cite{ducloscianciFastDecodersTopological2010d,bravyiAnalyticNumericalDemonstration2013,hutterImprovedHDRGDecoders2015a}.

For quantum LDPC codes, several successful decoders combine local message passing with nonlocal post-processing procedures~\cite{roffeDecodingQuantumLDPC2020a,panteleevDegenerateQuantumLDPC2021b}. Notable exceptions include local iterative decoders based on expansion, such as the small-set flip algorithm and related constructions~\cite{leverrierQuantumExpanderCodes2015,fawziEfficientDecodingRandom2018,guEfficientDecoderLinear2022,dinurGoodQuantumLDPC2022,leverrierDecodingQuantumTanner2022}. Aside from locality, many of these methods are typically implemented in a non-autonomous, ``decode-then-correct'' fashion: syndrome data are first collected, a classical algorithm is run to completion, and only then is the inferred correction applied. This separation between inference and feedback contrasts with an \emph{autonomous} decoder, in which local processing and corrective dynamics proceed concurrently.

A natural starting point for local decoders is the belief propagation (BP) algorithm~\cite{kschischangFactorGraphsSumProduct}. BP corresponds to a family of message passing algorithms in which information is exchanged between neighbors of a Tanner graph, and has played a central role in the success of classical LDPC codes~\cite{gallager1962low,tanner1981recursive,mackay1996near,richardson2001capacity}. These methods have also been adapted quite extensively to quantum codes, as far back as~\cite{mackaySparseGraphCodes2004a}. Their performance in the quantum setting, however, is often substantially limited by features absent from the classical counterparts, most notably the extensive error degeneracy and the resulting trapping sets of standard BP~\cite{poulinIterativeDecodingSparse2008,raveendranTrappingSetsQuantum2021a}. Consequently, successful decoding with BP often requires augmenting BP with additional, and frequently nonlocal decoding steps~\cite{ducloscianciFastDecodersTopological2010d,roffeDecodingQuantumLDPC2020a, panteleevDegenerateQuantumLDPC2021b,yeBeamSearchDecoder2025,mullerImprovedBeliefPropagation2025a,kuoExploitingDegeneracyBelief2022a,rigbyModifiedBeliefPropagation2019,grospellierNumericalStudyHypergraph2019,crestStabilizerInactivationMessagePassing2023,crestCheckAgnosiaBasedPostProcessor2024,gongLowlatencyIterativeDecoding2024a,hillmannLocalizedStatisticsDecoding2024,yinSymBreakMitigatingQuantum2024,higgottImprovedDecodingCircuit2023a,wolanskiAmbiguityClusteringAccurate2025}

To align with the goal of realizing error correction as a local non-equilibrium process, we take the following \emph{distributed computing} approach to a local decoder. We envision classical processors laid out according to the Tanner graph of the quantum code, with the only allowed operations being: local measurement of syndrome information, communication between neighboring processors, ultimately followed by local feedback back onto the quantum processor. This viewpoint is closely related to the existing cellular-automaton decoders present in the literature~\cite{harrington2004analysis,heroldCellularautomatonDecodersTopological2015a,breuckmannLocalDecoders2D2016,heroldCellularAutomatonDecoders2017,kubicaCellularautomatonDecodersProvable2019,vasmerCellularAutomatonDecoders2021,guedesQuantumCellularAutomata2024,palettaHighperformanceLocalDecoders2025b,lakeFastOfflineDecoding2025,lakeLocalActiveError2025a,balasubramanianLocalAutomaton2D2026,selubLocalDecodersFaulttolerant2026} where the decoding is re-cast as a dynamical process in its own right.

Our main result is a local measurement-and-feedback architecture for inference on quantum codes. A series of classical processors (see \cref{fig:fig_setup}) exchange messages amongst themselves according to belief propagation. Each check node is additionally equipped with an independent Poisson clock: when the clock of an active check rings, the check is allowed to \emph{move} by using the local beliefs for direction. In the toric code, these updates induce a local dynamics in which point-like defects move, pair, and annihilate. The same construction, however, is defined directly on the Tanner graph and therefore extends naturally to more general quantum LDPC codes.

The key departure from conventional BP decoding is that we do not require the message-passing dynamics to produce a globally consistent correction. Standard BP is typically iterated either for a prescribed number of rounds or until a convergence criterion is met, after which a global hard decision is formed. When termination is based on convergence, certifying that the decoder has finished requires information to be aggregated across the code and is therefore itself a nonlocal operation. Moreover, for the toric code, this conventional use of BP does not exhibit a decoding threshold~\cite{roffeDecodingQuantumLDPC2020a}. By contrast, we only use BP as a local inference engine. Individual marginals can provide reliable local information even when the collection of beliefs is not globally consistent. These local inferences are converted immediately into physical feedback, which updates the syndrome and allows the inference process to continue dynamically.

The paper is organized as follows. We first define a discrete-time inference machine, which takes conventional BP and turns it into an autonomous dynamics. This realizes an inference machine for any code already exhibiting a threshold under ordinary BP. For concreteness, we prove its performance on the one-dimensional repetition code. We then turn to the toric code, where we emphasize a key conceptual distinction: rather than using local message passing to reconstruct a globally consistent solution, we use it only to distribute reliable local marginal information. This motivates a continuous-time inference machine in which active syndrome defects use the local beliefs to move asynchronously as dictated by independent Poisson clocks. We show numerically that the resulting dynamics exhibits threshold behavior in the point-like sectors of both the two- and three-dimensional toric codes. Finally, exploiting the Tanner-graph generality of BP, we extend the construction beyond these settings and study its performance on the membrane-like sector of the three-dimensional toric code and on a family of bivariate-bicycle quantum LDPC codes.

\textit{Preliminaries.---}We work with $[[n,k,d]]$ CSS quantum codes on qubits described by parity-check matrices 
\[
H_X\in\Ftwo ^{m_X\times n},
\qquad
H_Z\in\Ftwo ^{m_Z\times n},
\]
satisfying the commutation conditions $H_Z H_X^\top = 0$ and work throughout in the code-capacity setting with perfect-syndrome measurements. For concreteness, we deal with $Z-$only noise. The corresponding inference problem is naturally represented on a Tanner graph $\cT=(\cV \sqcup \cC, \cE)$, whose variable nodes $v\in \cV$ represent data qubits and the check nodes $\cC$ represent $X-$type stabilizer measurements with $\{v,c\} \in \cE$ iff $[H_X]_{cv}=1$. 

Given a syndrome $S \subset \cC$  produced by the ground-truth error chain $E_g\subset\cV$~\footnote{We use capitalized $E\subset\cV$ for subset notation and $e \in \Ftwo^n$ to denote bit-vectors.}, the decoding task is to infer an error $E \subset \cV$ that is, 
\begin{enumerate}
    \item[(1)] \emph{Consistent.} That is, it produces the same syndrome $H_X e = s$, and 
    \item[(2)] \emph{Homologically equivalent.} Differing from the true error only by a stabilizer $e\oplus e_g\in\operatorname{row}(H_Z)$ 
\end{enumerate}
If either (1) or (2) is not true, we have a decoder failure.

The procedure of belief propagation can be expressed by writing down the ``partition function,'' for a given syndrome $s$,
\[
\Pr(s) := \sum_{e: H_X e = s} q^{|e|}(1-q)^{n-|e|}
\]
as a tensor network with bond dimension 2, and then performing message-passing over it~\cite{TNBP,TNMP}. The formal construction is described in the appendix. Working directly with log-likelihoods, each edge $\{c,v\}\in \cE$ is assigned two \emph{messages}, $m_{c\to v}, m_{v\to c}\in \R$, and the message-passing procedure consists of performing the following update iteratively,
\begin{align}\label{eq:BP-LLR1}
     m_{v \to c} &\leftarrow \pi_v + \sum_{c' \in \cN(v) / c}m_{c' \to v} \\  \label{eq:BP-LLR2}
    m_{c\to v} &\leftarrow 2(-1)^{s_c + d_c} \atanh\left(\prod_{w\in\cN(c)/v}\tanh\left[\frac{m_{w\to c}}{2}\right]\right)
\end{align}
where, under an i.i.d. error model with decoder prior $q$, we have $\pi_v = \log{q/(1-q)}$, and $d_c := |\cN(c)|$. Denoting the combined message state as $m$ and the update map as $F_{\BP}$, the iteration is compactly written as $m^{(i+1)} = F_{\BP}(m^{(i)})$. Given a set of messages $m$, one obtains an approximation to the local bit beliefs $\Pr(e_v|s)  \approx \Pr_{\BP}(e_v|s)$, with
\[
 \log{\frac{\Pr_{\BP}(e_v|s)}{1-\Pr_{\BP}(e_v|s)}} \equiv L_v := \pi_v + \sum_{c\in\cN(v)}m_{c\to v}
\]
The traditional paradigm of BP based decoders proceeds as follows. One runs the BP iterations for a fixed number of rounds. From the available messages, one then outputs an error candidate $e$ defined bitwise by as $e_{v}:= \argmax_{b\in\{0,1\}} \Pr_{\BP}(e_v=b|s)$. If the error has ``converged,'' viz., $H_X e=s$, one outputs $e$. This procedure is exact on inference problems on trees, and an approximation otherwise.

\begin{figure}
    \centering
    \includegraphics[width=1.0\linewidth]{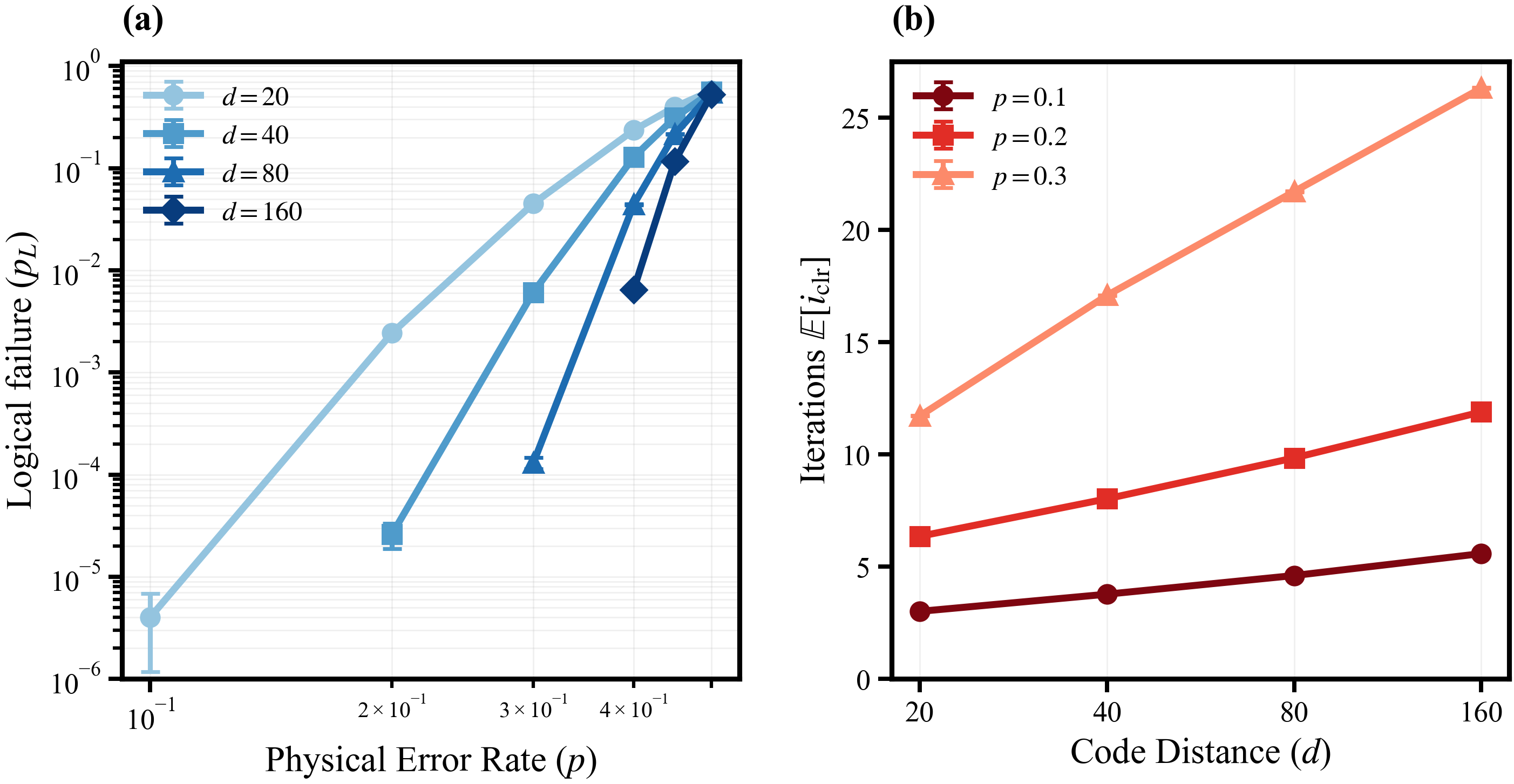}
    \caption{\textbf{Discrete local inference dynamics.} (a) Logical error rates $p_L$ as a function of physical error rate $p$ for the repetition code with distances $d\in \{20,40,80,160\}$. (b) Average clearing iteration $\mathbb{E}[i_{\text{clr}}]$ as a function of code distance $d$.}
    \label{fig:repcode}
\end{figure}

\textit{Local inference machines.---}We now proceed to construct a local inference machine, comprising local measurement, communication and subsequent feedback. We envision an array of classical processors, one for each of a data and measurement qubit, laid out according to the geometry of the device. This is illustrated for the repetition code in~\cref{fig:fig_setup}. Each individual processor $x$ carries the state $\xi_x$, which includes the outgoing message $m_{x\to y}$ for every neighbor $y$. In addition, a check processor $c\in\cC$ stores its local syndrome value $s_c$, and each data processor $v\in \cV$ stores its prior $\pi_v$. The state of a processor is updated using only its own state, and those of its immediate neighbors, implementing the BP equations of \cref{eq:BP-LLR1,eq:BP-LLR2}. Thus, instead of running BP as a standalone iterative algorithm and subsequently forming global decisions, we regard the distributed message dynamics as the inference layer of an autonomous dynamics---the local corrections continuously modify the syndrome on which inference is performed. 

The inference dynamics is equipped with local \emph{moves}, wherein a data processor decides whether to apply feedback. Concretely, at iteration $i$ the following updates occur from the data of $i-1$: the processor $v\in\cV$ reads the state of its neighbors, and forms the feedback bit $f_v^{(i)} := \1[L^{(i)}_v > 0]$ where 
\begin{equation}\label{eq:llrupdate}
    L^{(i)}_v \leftarrow \pi_v^{(i-1)} + \sum_{c\in\cN(v)}m^{(i-1)}_{c\to v}
\end{equation}
That is, a data processors decides to apply feedback if its local belief satisfies $\Pr_{\BP}(e_v=1|s) > 0.5$. To account for nontrivial feedback on bit $v$, all of its corresponding messages and its prior are multiplied subsequently by a gauge factor of $(-1)^{f^{(i)}_v}$. Every check processor $c\in\cC$ then updates its syndrome register, 
\begin{equation}\label{eq:synupdate}
    s_c^{(i)} \leftarrow s_c^{(i-1)} \oplus \bigoplus_{v\in\cN(c)}f_v^{(i)}
\end{equation}
Crucially, the dynamics of \cref{eq:llrupdate,eq:synupdate} are \emph{autonomous}, no global ``overseer'' checks, e.g., if the messages have converged. This way, any code obtaining a threshold under ordinary BP can attain the same threshold using the discrete autonomous machine. For concreteness, we proceed to prove the performance of the inference dynamics on the 1D repetition code. This will be instructive in setting out expectations of the average time taken to decode with the machine. The clearing iteration $i_{\text{clr}}$ is defined to be the time to clear the syndrome $\min_{i \geq 0} s^{(i)} = \bm 0$.

\begin{maintextthm}[Local inference on the rep code]\label{maintextthm:rep}
    The discrete inference machine achieves the optimal threshold $p_{\thr} = 0.5$ with an average clearing time of $\E[i_{\clr}] = O(\log d)$ for a distance $d$ rep code. 
\end{maintextthm}

The main technical ingredient of this proof is an exact mapping of the local inference dynamics onto ordinary belief propagation, which is known to be exact in one dimension (and more generally, on trees). The rest follows by using standard ``clustering'' properties of the i.i.d. noise distribution. In particular, clusters of diameter $w$ get cleared in $O(w)$ time. The $O(\log d)$ clearing time then follows from the fact that one only needs to clear upto log-sized clusters to attain threshold performance~\cite{kovalevFaultToleranceQuantum2013}. The formal proof can be found in the SM. We numerically demonstrate the decoder in \cref{fig:repcode} for distances $d\in \{20,40,80,160\}$. We note that the LER curves intersect at $p=0.5$, and the average clearing time follows a logarithmic trend.

We now ask whether the same strategy can be extended to quantum codes, taking the toric code as the simplest nontrivial setting. An obstacle is immediate: standard belief propagation does not exhibit a threshold on the toric code~\cite{roffeDecodingQuantumLDPC2020a}. It is instructive to understand this failure, as it exposes a basic tension between local marginal information computed by BP and and demanding a globally consistent solution. To see this, consider an elementary plaquette (far away from other defects) of the toric code with qubits on the edges and checks on the vertices, and consider a syndrome on diagonally opposite corners. 
\[
\LoopTN
\]
Now, even if an algorithm can provide local marginals, the individual posterior state $\Pr(e|s)$ of all the qubits is (nearly) fully mixed. Thus, performing a simultaneous hard decision will not yield a consistent error. However, the BP algorithm can provide the degenerate directions the checks need to \emph{move} in, to annihilate each other. This is our key observation, instead of asking a local algorithm to provide globally consistent solutions, we merely demand \emph{locally} consistent moves. This also brings out a necessary change in perspective going from the rep code to the toric code: the feedback rule can no longer be bit-centric owing to the exponential degeneracy; instead, it must be \emph{check-centric}, wherein an active check determines whether to move or not.

\begin{algorithm}[t]
\caption{Local Inference Machine}
\label{alg:localmachine}
\begin{algorithmic}[1]

\Require Parity-check matrix $H_X$, initial syndrome $S(0)$
\State Initialize message state $m(0):=0$ and independent unit-rate Poisson clocks $\{N_c(t)\}_{c\in\mathcal C}$ 

\Statex
\While{decoding is active}

    \State Evolve the continuous-time BP dynamics
    \[
        \tau \frac{dm}{dt}
        =
        F_{\mathrm{BP}}\!\left(m;S(t)\right)-m .
    \]

    \State Let $c$ be the next check whose Poisson clock rings

    \If{$c \in S(t)$}
    \If{\textsc{Evidence}$(c) \ge 1/d_c$}    
    \State
        \[
            f_c
            \gets
            \textsc{LocalMove}\!\left(c,\xi(t)\right)
        \]
        \State Apply the local correction $e(t^+) \gets e(t) \oplus f_c$
        \State Update the syndrome
        \[
            s(t^+)
            \gets
            s(t)\oplus H_X f_c .
        \]
    \EndIf
    \EndIf

\EndWhile

\end{algorithmic}
\end{algorithm}

\emph{Continuous inference machines.---}We now promote the inference dynamics to continuous time, and shift the task of local decisions from data processors onto the check processors. The processors have local states $\xi$ exactly as before, except now the BP equations are carried out in continuous time (see~\cite{hematiDynamicsPerformanceAnalysis2006} for a similar analog implementation in a different context):  
\begin{equation}
     \tau \frac{dm}{dt}
        = F_{\mathrm{BP}}\!\left(m;s(t)\right)-m .
\end{equation}
where $m$ is the combined state of the messages, and $F_{\mathrm{BP}}$ is the BP update, and $\tau$ is the relaxation time constant. Then, each syndrome processor is equipped with a local Poisson clock, independent of other processors. This provides a mechanism for local feedback: whenever a processor's clock rings, it attempts to \emph{move} using the local beliefs of its neighboring checks as direction. This architecture is summarized in \cref{alg:localmachine}.

Let $e(0)$ denote the initial error on the device, resulting in the syndrome $s(0)$. Then, each of the check processors $c$ loads in $s_c(0)$, and the dynamics begins. With the continuous inference dynamics running, local moves proceed as follows. Let $N_c(t)$ denote the local clock of processor $c\in\cC$. When the clock on check $c$ rings, i.e., $N_c(t) = 1$, it is allowed to perform a \textsc{Localmove} iff,
\begin{enumerate}
    \item[(i)] it is active, $s_c(t) = 1$, and, 
    \item[(ii)] it has evidence, i.e., $\textsc{Evidence}(c) \ge 1/d_c$ where 
    \[
  \textsc{Evidence}(c)  :=  \max_{v\in\cN(c)}\Pr_{\BP}(e_v|s)
    \]
    and $d_c$ is the degree of check $c$.
\end{enumerate}
The evidence criteria comes from a simple union bound: for a degree-$d_c$ active check $c\in S$, we have $\max_{v\in \cN(c)} \Pr(e_v|s) \geq 1/d_c$. This serves as a confidence metric: the check does not move until the surrounding local beliefs have accumulated sufficient signal. The action of the \textsc{Localmove} at an active check $c$ is to flip the qubit $v\in \cN(c)$, where 
\begin{equation}
    v \in \argmax_{w\in \cN(c)} \Pr_{\BP}(e_w|s),
\end{equation}
breaking ties randomly. Then, the feedback at time $t$ is $f(t) = \1_{v}$ (note that almost surely, only one clock rings at $t$). One then writes the cumulative feedback upto time $t$ and residual error at time $t$, respectively:
\begin{equation}
    z(t) := \bigoplus_{u < t}f(u), \qquad e(t) := e(0) \oplus z(t).
\end{equation}
We define the clearance time to be:
\begin{equation}
    t_{\clr} := \min\{t~:~ H_X[e(t)] = 0\}
\end{equation}
and $t_{\clr} := \infty$ if the syndrome does not clear.

We begin with toric codes: qubits are placed on the edges of a square lattice, and $X-$checks are placed on the vertices. Herein, errors create pairs of point-like excitations on the vertices, and decoding can be viewed as the task of pairing and annihilating these defects. The \textsc{localmove} provides a local realization of this pairing picture: when an active defect attempts to move, the locally inferred marginals determine the direction, progressively bringing defects together until they annihilate. In this sense, the dynamics replaces a global matching computation by a sequence of local and belief-guided moves. Given faithful access to local marginals, one expects that an isolated error-cluster of diameter $w$ gets cleared by this dynamics in $O(w)$ time, except with a probability atmost $O(e^{-w})$. We study this rigorously in a companion work~\cite{rigorousinfmachine}.

For BP, the picture is complicated by the presence of short loops in the Tanner graph. This is due to the known ``self-feedback" issue on loopy graphs, wherein the signal from a distant defect is obfuscated by the feedback around a loop. For a Tanner graph with girth~$g$, the local computation remains tree-like upto distance $O(g)$. As a result, we expect the cluster annihilation dynamics to hold for clusters upto diameter $O(g)$. Since threshold performance requires clearing out clusters only upto diameter $O(\log d)$~\cite{lakeFastOfflineDecoding2025,bravyiAnalyticNumericalDemonstration2013}, this suggests that the continuous time machine will operate well upto code sizes $d \leq e^{O(g)}$. As we show below, our numerical results probe precisely this finite-size regime and exhibit robust threshold-like behavior.

\begin{figure}
    \centering
    \includegraphics[width=1.0\linewidth]{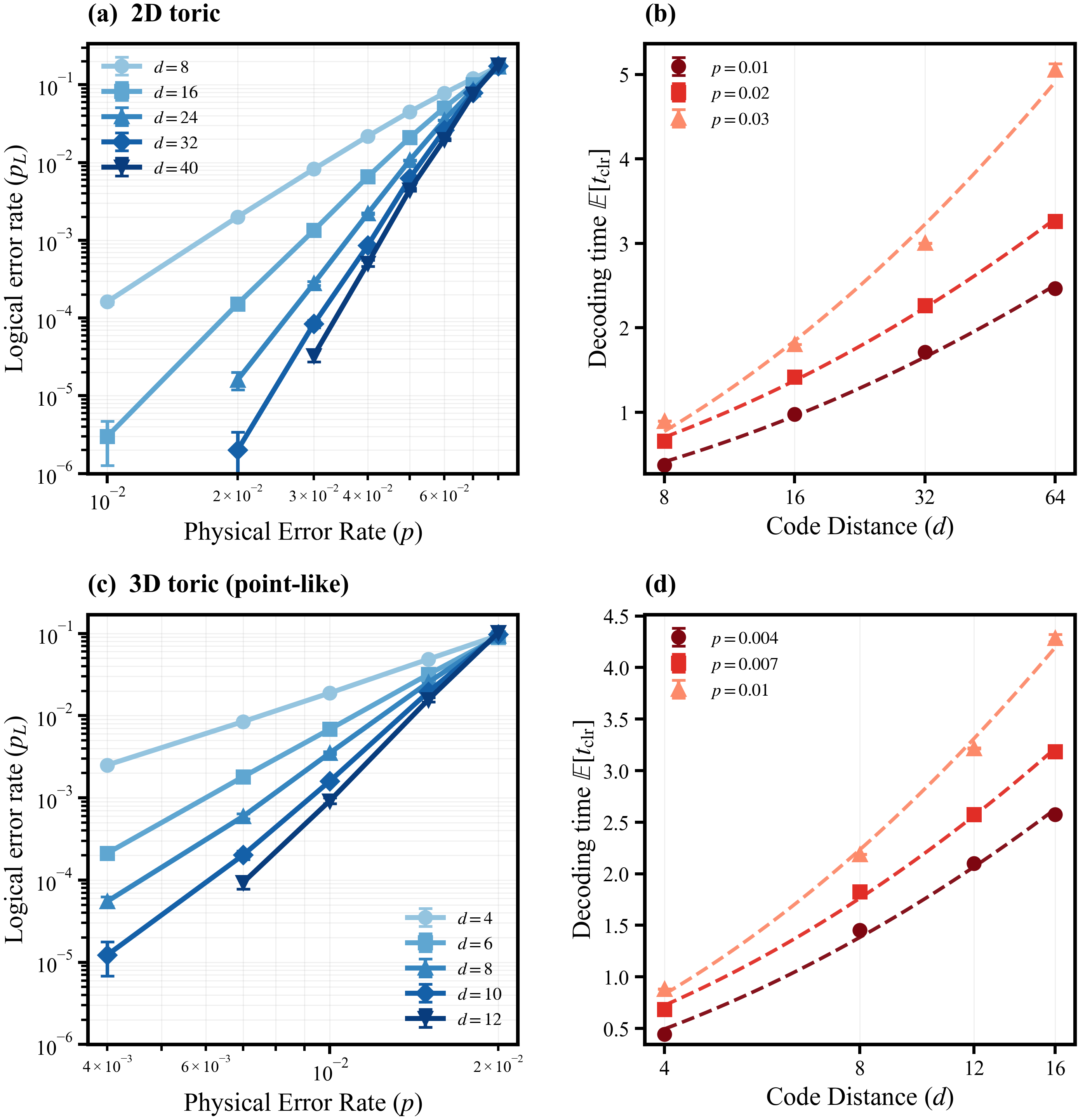}
    \caption{\textbf{Local inference on toric codes.} Logical error rates and expected clearing time on the (a-b) two-dimensional toric code and (c-d) point-like sector of the three-dimensional toric code. The fits are of the form $\E[t_{\clr}]= a_p (\log d)^2 + b_p$ for $p-$dependent constants $a_p,b_p$.}
    \label{fig:toric}
\end{figure}

The performance of the continuous-time inference machine on the two-dimensional toric code is shown in \cref{fig:toric} (using $\tau=0.1$ in the unit of rate-one clocks in all numerics hereon). In panel (a), the finite-size crossings are consistent with a decoding threshold of approximately $p_{\mathrm{th}}\simeq 8\%.$ The expected decoding time is shown in panel (b), with a fit of $a_p \cdot (\log d)^2 + b_p$. We now consider the 3D toric code, in the so-called ``point-like" sector: qubits are placed on the edges of a 3D cubic lattice, with $X-$checks on the vertices. Using effectively now a 3D layout of the local processors, we study the performance of the machine in \cref{fig:toric}(c-d). The resulting logical error rates are shown in \cref{fig:toric}(c), where the finite-size crossings are consistent with a threshold of approximately $p_{\mathrm{th}}\simeq 2\%.$ The expected clearing times are shown in panel (d) with a fit of $a_p \cdot (\log d)^2 + b_p$.

\begin{figure}
    \centering
    \includegraphics[width=1.0\linewidth]{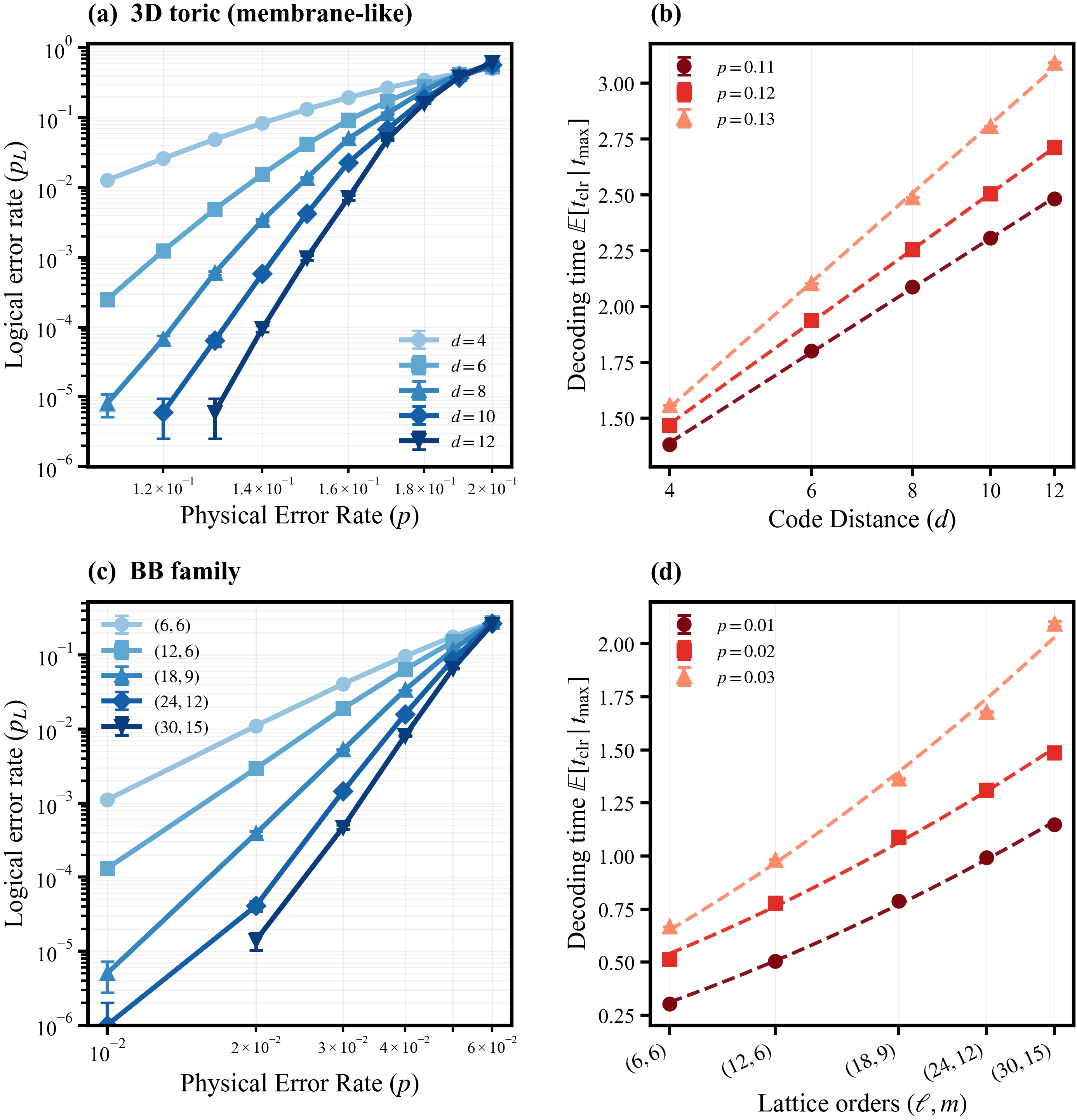}
    \caption{\textbf{Local inference on LDPC codes.} Logical error rates and expected clearing time (for convergence in $t \leq t_{\max} = 1000$) for the (a-b) membrane-like sector of the three-dimensional toric code with a fit $\E[t_{\clr}]= a_p (\log d) + b_p$, and (c-d) a family of bivariate-bicycle codes including the Gross code with a fit $\E[t_{\clr}]= a_p (\log \sqrt{2\ell m})^2 + b_p$. }
    \label{fig:ldpc}
\end{figure}

We next turn to quantum LDPC codes whose decoding problems are not naturally reducible to pair-wise defect matching. We focus on codes with geometrically local stabilizers in two or three dimensions, but for which a single-qubit error can excite more than two checks, so that the variable degree satisfies $\Delta_v>2$. Somewhat surprisingly, we find that the same \textsc{localmove} remains effective. As an alternative, one can also consider a local move consisting of a \emph{multi-qubit} flip, akin to the small-set flip~\cite{leverrierQuantumExpanderCodes2015} algorithm however guided by the local beliefs. We find similar performance for both, and report the single-qubit moves in this work.

We first consider the 3D toric code in the ``membrane-like'' sector, where the qubits are again on the edges, with checks now on the faces of the 3D lattice. In this sector, each data qubit participates in $\Delta_v=4$ relevant parity checks, where a membrane-like error creates a loop-like violation of parity-checks. We show in~\cref{fig:ldpc}(a) the resulting logical error rates for code distances \(d\in\{4,6,8,10,12\}\). The crossings are consistent with a finite threshold of approximately $p_{\mathrm{th}}\simeq 19\%.$ Unlike the point-defect sectors of the toric code, the syndrome need not clear for every realization in these more general codes. We therefore impose a maximum decoding time $t_{\max}$. A sample is declared nonconvergent if its syndrome has not cleared by $t_{\max}$, and such events are counted as decoding failures when computing the logical error rate. To characterize the convergence time of successful runs, we report the mean clearing time conditioned on convergence within the allotted window, abbreviated $\E[t_{\clr}|t_{\max}]$. This is shown for the membrane sector of the 3D toric code in ~\cref{fig:ldpc}(b), with a plotted fit of $a_p \log d + b_p$.

As a more general test of the Tanner-graph formulation, we next consider a family of bivariate-bicycle (BB) codes~\cite{bravyiHighthresholdLowoverheadFaulttolerant2024b}. A BB code on an $(\ell,m)$ torus consists of qubits on the edges of a torus, with checks on vertices and faces. The check connectivity, however, is specified by two polynomials over $\Ftwo[x,y]/(x^\ell-1,y^m-1)$\footnote{As a simple comparison, the polynomials for the 2D toric code can be written as $A=1+x, B=1+y$}, see~\cite{sahayMatchingDecoderBivariate2026} for a nice exposition. We fix the polynomial pair used for the Gross code,
$$
A=x^3+y+y^2,\qquad
B=y^3+x+x^2,
$$

with corresponding CSS check matrices $H_X=[A\mid B]$ and $H_Z=[B^\top\mid A^\top]$. We retain the same local polynomial structure while increasing the torus dimensions, considering $(\ell,m)\in\{(6,6),(12,6),(18,9),(24,12),(30,15)\}$. With the fixed polynomial, an increasing sequence of such codes can be mapped to a stack of toric codes~\cite{haahClassificationTranslationInvariant2021,bombinStructure2DTopological2014,bombinUniversalTopologicalPhase2012a,wang2026decoupling2dtranslationinvarianttopological}, and we expect the sequence of distances to be $O(\sqrt {\ell m}$). The first two members are respectively the $[[72,12,6]]$ half-Gross code and the $[[144,12,12]]$ Gross code~\cite{bravyiHighthresholdLowoverheadFaulttolerant2024b}. The logical error rates for this sequence are shown in \cref{fig:ldpc}(c), where the finite-size crossings indicate a threshold near $p_{\mathrm{th}}\simeq 6\%$. We show the expected clearing time, again with a cutoff of $t_{\max} = 1000$ in \cref{fig:ldpc}(d) with a fit of $a_p (\log \sqrt{2\ell m})^2 + b_p$, where we use the square-root of the blocklength $n=2\ell m$ as a proxy for distance.

\emph{Discussion.---}We have constructed a distributed model of inference comprising of local measurement, computation and feedback using a series of classical processors laid out in the geometry of the Tanner graph of a quantum CSS code. Taking any code already exhibiting a threshold under ordinary BP, we show how to construct an autonomous inference machine which preserves the threshold. Turning to quantum codes not decodable under ordinary BP, we show that, under a suitable change of viewpoint from demanding a globally consistent solution onto only a locally optimal direction, one can construct local inference machines applicable to quantum LDPC codes. We tested this machine on the two-dimensional and three-dimensional toric codes, as well as a family of bivariate-bicycle codes and numerically showed threshold performance.

The results presented here concern \emph{offline} decoding: an error is sampled at $t=0$, after which the inference dynamics evolves in the absence of further noise. A natural next step is to extend the framework to \emph{online} decoding, in which new faults and syndrome information continuously arrive while inference and feedback proceed in real time. A natural route is to promote the Tanner graph to a spacetime decoding graph and run the local inference dynamics within overlapping spacetime windows, whose extent need only capture the relevant growing error clusters~\cite{skoricParallelWindowDecoding2023,tanScalableSurfaceCodeDecoders2023,lakeLocalActiveError2025a}. The intrinsically distributed architecture considered here appears particularly well suited to this setting, since inference and feedback can proceed continuously without waiting for a global decoding decision.

A second direction is to improve the local inference dynamics itself. Throughout this work we use standard BP without code-specific tuning, leaving considerable freedom in how local posterior information is generated and propagated. Many modified message-passing schemes include, e.g., memory, damping, different update schedules, or decimation~\cite{kuoExploitingDegeneracyBelief2022a,oldGeneralizedBeliefPropagation2023b,yaoBeliefPropagationDecoding2024b,mullerImprovedBeliefPropagation2025a} resulting in improved performance. Incorporating such methods within the local feedback architecture developed here may improve both the threshold and the convergence time, while preserving the essential distributed character.

Lastly, the continuous evolution of the syndrome in a dynamic inference network can fit naturally into a hierarchical decoding architecture. In such schemes, a fast local decoder quickly removes simple error configurations, followed up by a slower global decoder invoked only on the residual syndrome~\cite{delfosseHierarchicalDecodingReduce2020,smithLocalPredecoderReduce2023,chamberlandTechniquesCombiningFast2023,cauneBeliefPropagationPartial2023}. The inference dynamics studied here might be a natural candidate for the first stage of such an approach, ideally resulting to a call for the global stage only for atypical errors.

\emph{Acknowledgements.---}We thank Sarang Gopalakrishnan and Frank Zhang for helpful discussions. This work was supported by a Brown Investigator Award, a program of the Brown Institute for Basic Sciences at the California Institute of Technology.

\emph{AI Statement.---}Large language models (in particular, GPT 5.6 Sol and Claude Opus 5.5) were used to assist with the numerical codes, literature search, and technical proofs.


\bibliographystyle{apsrev4-2}
\bibliography{refs}


\clearpage
\onecolumngrid
\makeatletter
\let\@outputdblcol\@outputpage
\makeatother
\setcounter{section}{0}
\setcounter{subsection}{0}
\setcounter{subsubsection}{0}
\setcounter{equation}{0}
\setcounter{figure}{0}
\setcounter{table}{0}
\setcounter{footnote}{0}
\renewcommand{\thesection}{S\arabic{section}}
\renewcommand{\theHsection}{S\arabic{section}}
\renewcommand{\thesubsection}{\thesection.\arabic{subsection}}
\renewcommand{\theequation}{S\arabic{equation}}
\renewcommand{\thefigure}{S\arabic{figure}}
\renewcommand{\thetable}{S\arabic{table}}
\renewcommand{\theHequation}{S\arabic{equation}}
\providecommand{\theHfigure}{}\renewcommand{\theHfigure}{S\arabic{figure}}
\providecommand{\theHtable}{}\renewcommand{\theHtable}{S\arabic{table}}
\section*{Supplementary Material}
\listofsmentries

\smsection{Tanner Graphs, Tensor Networks and Belief Propagation}

\begin{defn}[Tanner Graph]
    Let $H_X\in \mathbb F_2^{m_X\times n}$ be the parity-check matrix of concern.
Associated to $H_X$ is its Tanner graph $G_H=(\cV\sqcup \cC,\cE),$ which is a bipartite graph with data vertices $\cV=\{1,\dots,n\}$ and syndrome, or check, vertices $\cC=\{1,\dots,m\}.$ There is an edge between data vertex \(v\in\cV\) and syndrome vertex
\(c\in\cC\) iff $H_{cv}=1.$ 
\end{defn}
For \(c\in\cC\), the neighborhood \(\cN(c)\subseteq \cV\)
is the set of data bits participating in check \(c\). Similarly, for
\(v\in\cV\), the neighborhood \(\cN(v)\subseteq \cC\)
is the set of checks incident on bit \(v\). We denote, 
\begin{align}
    \text{(Check degree) } \Delta_c &:= \max\{|\cN(c)| : c\in \cC\} \\ 
    \text{(Variable degree) } \Delta_v &:= \max\{|\cN(v)| : v\in \cV\} 
\end{align}
We say that the code is LDPC if $\Delta_c, \Delta_v = O(1)$ and do not grow with the size of the code.

We write the error configuration as $ e=(e_v)_{v\in\cV}\in \mathbb F_2^n$
and assume an i.i.d. Bernoulli prior
\[
    \Pr(e)
    =
    \prod_{v\in\cV} q^{e_v}(1-q)^{1-e_v}.
\]
The syndrome associated with an error configuration is $ s=H_X(e).$ Equivalently, for each individual check \(c\in\cC\),
\[
    s_c
    =
    \bigoplus_{i\in\cN(c)} e_v.
\]
\smsubsection{The inference tensor network}
Given an observed syndrome $S$, the posterior distribution over errors is
\begin{equation}
    \Pr(e\mid s)
    =
    \frac{\Pr(e,s)}{\Pr(s)}
    =
    \frac{\Pr(s\mid e)\Pr(e)}{\Pr(s)}.
\end{equation}
Since the syndrome is a deterministic function of the error, we have $\Pr(s\mid e) = \mathbf 1[H_X e = s].$ Therefore under an i.i.d. prior with single qubit error probability $q$,
\begin{equation}
    \Pr(e\mid s)
    =
    \frac{1}{\Pr(S)}
    \mathbf 1[H_X e = s]
    \prod_{v\in\cV} q^{e_v}(1-q)^{1-e_v}.
\end{equation}
Equivalently, writing the syndrome constraints locally,
\begin{equation}
    \Pr(e\mid s)
    =
    \frac{1}{Z(s)}
    \prod_{c\in\cC}
    \mathbf 1\!\left[
        \bigoplus_{w\in\cN(c)} e_w = s_c
    \right]
    \prod_{v\in\cV} q^{e_v}(1-q)^{1-e_v},
\end{equation}
where
\begin{equation}
    Z(s)=\Pr(s)
\end{equation}
is the syndrome partition function. We now define a tensor network whose full
contraction computes this normalization.

\begin{defn}[Inference Tensor Network]
Fix a parity-check matrix $H_X\in\mathbb F_2^{m_X\times n}$, a decoder prior $q\in[0,1]$, and a syndrome $s\in\mathbb F_2^m$. Let $\cT=(\cV\sqcup\cC,\cE)$ be the Tanner graph of $H_X$. The associated inference tensor network
$\mathsf{TN}(H_X,s)$ is defined as follows. For each Tanner-graph edge $\{c,v\}\in\cE$, we introduce a binary virtual index $x_{cv}\in\{0,1\}.$ For each syndrome vertex $c\in\cC$, define a syndrome tensor imposing the parity constraint:
\[
    T^{(c,s_c)}_{\{x_{cv}\}_{v\in\cN(c)}}
    :=
    \mathbf 1\!\left[
        \bigoplus_{v\in\cN(c)} x_{cv}=s_c
    \right].
\]
For each data vertex \(v\in\cV\), define a data tensor
\[
    D^{(v)}_{\{x_{cv}\}_{c\in\cN(v)}}
    :=
    \sum_{e_v\in\{0,1\}}
    \operatorname{pr}_v(e_v)
    \prod_{c\in\cN(v)}
    \delta_{x_{cv},e_v},
\]
where
\[
    \operatorname{pr}_v(e_v)
    :=
    q^{e_v}(1-q)^{1-e_v}.
\]

\end{defn}

The two types of tensors present in the network can be drawn as follows. The check tensors enforce the parity of the corresponding syndrome: 
\begin{align}
    \ParityPic{s_c}
    \quad &\equiv \quad
    T^{(c,s_c)}_{\{x_{cv}\}_{v\in\cN(c)}}
    \notag\\[1mm]
    &:= 
    \mathbf{1}\!\left[
        \bigoplus_{v\in\cN(c)} x_{cv}=s_c
    \right],
    \qquad c\in\cC.
    \label{eq:check-tensor-def}
\end{align}

The data tensors enumerate the possible errors with the corresponding prior:
\begin{align}
    \DataPic
    \quad &\equiv \quad
     D^{(v)}_{\{x_{cv}\}_{c\in\cN(v)}}
    \notag\\[1mm]
    &:=
    \sum_{e_v\in\{0,1\}}
    \operatorname{pr}_v(e_v)
    \prod_{c\in\cN(v)} \delta_{x_{cv},e_v},
    \qquad v\in\cV.
    \label{eq:data-tensor-def}
\end{align}

where the prior tensor is given as,
\begin{align}
    \PriorPic := \begin{pmatrix}
        1-q\\[1mm]
        q
    \end{pmatrix},
    \qquad v\in\cV.
    \label{eq:prior-tensor-def}
\end{align}

\begin{prop}[Contraction equals the syndrome probability]
Let \(\mathsf{TN}(H_X,s)\) be the inference tensor network. Its contraction value equals
\begin{equation}
    Z_{\mathsf{TN}}(s)=\Pr(s),
\end{equation}
for the i.i.d. Bernoulli error prior
\[
    \Pr(e)
    =
    \prod_{v\in\cV} q^{e_v}(1-q)^{1-e_v}.
\]
\end{prop}

\begin{proof}
Substituting the definition of the tensors yields for the contraction,
\begin{align}
    Z_{\mathsf{TN}}(s)
    &=
    \sum_{\{x_{cv}\}}
    \prod_{c\in\cC}
    \mathbf 1\!\left[
        \bigoplus_{v\in\cN(c)}x_{cv}=s_c
    \right]
    \prod_{v\in\cV}
    \left(
        \sum_{e_v\in\{0,1\}}
        \operatorname{pr}_v(e_v)
        \prod_{c\in\cN(v)}\delta_{x_{cv},e_v}
    \right).
\end{align}
The delta tensors enforce
\[
    x_{cv}=e_v
    \qquad
    \text{for all } c\in\cN(v).
\]
Therefore the sum over virtual indices reduces to a sum over one error bit
\(e_v\) for each data vertex \(i\). Hence
\begin{align}
    Z_{\mathsf{TN}}(s)
    &=
    \sum_{e\in\mathbb F_2^n}
    \prod_{c\in\cC}
    \mathbf 1\!\left[
        \bigoplus_{v\in\cN(c)}e_v=s_c
    \right]
    \prod_{w\in\cV}\operatorname{pr}_w(e_w)
    \\
    &=
    \sum_{e\in\mathbb F_2^n}
    \mathbf 1[H_X e=s]\Pr(e)
    \\
    &=
    \Pr(S).
\end{align}
\end{proof}
\smsubsection{Belief propagation}
Belief propagation is a message-passing algorithm defined on top of this network, wherein each edge $\{c,v\}\in \cE$ is assigned two \emph{messages}, $\mu_{c\to v}, \mu_{v\to c}\in \R^2$, and the message-passing procedure consists of performing the following two updates,
\[
\CheckUpdatePic{S_c}\longrightarrow\mu_{c\to v}
\quad
\DataUpdatePic\longrightarrow\mu_{v\to c}
\]
where in each case the leg left open on the right is the free index carried by
the outgoing message, and the remaining neighbours are contracted against their
incoming messages. Denoting the combined message state as $\mu$ and the update map as $F_{\BP}$, the iteration is compactly written as $\mu^{(i+1)} = F_{\BP}(\mu^{(i)})$. Given a set of messages $\mu$, one obtains an approximation to the local bit marginals, 
\[
\Pr(e_v|s)  \approx \Pr_{\BP}(e_v|s) \quad := \DataMarginalPic
\]
We can equivalently write the procedure of belief propagation in terms of scalar log likelihoods. For each edge $\{c,v\}\in \cE$ we define the LLR space messages, 
\begin{equation}
    \underbrace{m_{c\to v} := \log\frac{\mu_{c\to v}(1)}{\mu_{c\to v}(0)}, \qquad m_{v\to c} := \log\frac{\mu_{v\to c}(1)}{\mu_{v\to c}(0)}}_{\text{LLR space messages}}, \qquad  \underbrace{\pi_v := \log{\frac{q}{1-q}}}_{\text{Prior LLRs}}
\end{equation}

In this formulation, the discrete BP equations read, 
\begin{align}
    \text{(Bit to check) } m_{v \to c}^{(i+1)} &:= \pi_v + \sum_{c' \in \cN(v) / c}m_{c' \to v}^{(i)} \\  
    \text{(Check to bit) } m_{c\to v}^{(i+1)} &:= 2(-1)^{s_c + d_c} \atanh\left(\prod_{w\in\cN(c)/v}\tanh\left[\frac{m_{w\to c}}{2}\right]\right)
\end{align}

\smsection{Autonomous machines}
We now analyze the discrete time local inference machine operating on the repetition code. All of the machinery here can be straightforwardly generalized onto a code where ordinary BP achieves a threshold, with appropriate modifications. We focus on the repetition code, state the algorithm precisely, and prove it's optimal performance. 
\smsubsection{Setup and algorithm}
We define the following parameters for the rest of this section: 
\begin{enumerate}
    \item[(1)] Initial error on the device, $e^{(0)}$
    \item[(2)] $e^{(i)}$ be the residual error on the device \emph{after} iteration $i$ is complete. 
    \item[(3)] $f^{(i)}$ be the feedback applied \emph{at the} step $i$ 
    \item[(4)] $z^{(i)}$ be the cumulative correction up until and including iteration $i$. We initialize $z^{(0)} = \bm 0$.
    \item[(5)] $s^{(i)} = H_X e^{(i)}$ be the syndrome stored in the network at the end of iteration $i$. 
\end{enumerate}

These obey the following elementary relations:
\begin{equation}
    z^{(i)} = \bigoplus_{j=1}^{i} f^{(j)}, \qquad e^{(i)} = e^{(0)} \oplus z^{(i)}.
\end{equation}
Moreover, for a check $c\in \cC$, with $\cN(c) = \{v,w\}$, 
\begin{equation}
s_c^{(i)} = s_c^{(0)} \oplus z_v^{(i)} \oplus z_w^{(i)}
\end{equation}
We choose a decoder prior $0 < q <1/2$ and set the initial messages and prior LLR to, 
\[
    \pi^{(0)} = \log{\frac{q}{1-q}}, \qquad m^{(0)}_{v\to c} = m^{(0)}_{c\to v} = 0.
\]
Every iteration $i$ consists of the following updates. 
\begin{enumerate}
    \item[(1)] \textbf{Check to bit message} For a check $c$ with $\cN(c) = \{v,w\}$, we have 
    \begin{equation}
        \hat{m}^{(i)}_{c\to v} \leftarrow (-1)^{s^{(i-1)}_c}m^{(i-1)}_{w\to c}
    \end{equation} 
    \item[(2)] \textbf{Bit to check message} For a bit $v\in \cV$ connected to a check $c\in\cC$ we have, 
    \begin{equation}
        \hat{m}^{(i)}_{v\to c} \leftarrow \pi^{(i-1)}_v + \sum_{c'\in \cN(v)/c}m^{(i-1)}_{c'\to v}
    \end{equation}
    \item[(3)] \textbf{Current beliefs}
    \begin{equation}
        \hat{L}_v^{(i)} := \pi_v^{(i-1)} + \sum_{c\in\cN(v)} m^{(i-1)}_{c\to v}
    \end{equation}
    Here the hats over $m$ and $L$ denote the variables \emph{before} the application of the local rule. 
    \item[(4)] The \textbf{Local rule} operates as follows:
    \begin{align}
        \text{(Feedback) } f^{(i)}_v &\leftarrow \1[\hat{L}_v^{(i)} > 0] \\ 
        \text{(Syndrome update) } s_c^{(i)}  &\leftarrow s_c^{(i-1)} \oplus \bigoplus_{v\in \cN(c)} f^{(i)}_v \\ 
        \text{(LLR switch) } \pi^{(i)}_v &\leftarrow (-1)^
        {f^{(i)}_v} \pi^{(i-1)}_v \\ 
        m_{v\to c}^{(i)} &\leftarrow (-1)^{f^{(i)}_v} \hat{m}_{v\to c}^{(i)} \\ 
        m_{c\to v}^{(i)} &\leftarrow (-1)^{f^{(i)}_v} \hat{m}_{c\to v}^{(i)} 
    \end{align}
\end{enumerate}

\smsubsection{Mapping to ordinary BP}
Let us first study a property of ``ordinary'' BP dynamics, which holds a fixed syndrome and runs until convergence. Denote the messages, priors, and bit beliefs of this dynamics with a bar as $\bar{m}$, $\bar{\pi}$, and $\bar{L}$ respectively. Consider a vector $g\in \Ftwo^n$, and define its gauge action on the BP variables as: 
\begin{align}
    (G_g \pi)_v &:=  (-1)^{g_v}\pi_v \\ 
    (G_g m)_{v\to c} &:=  (-1)^{g_v}m_{v\to c} \\ 
    (G_g m)_{c\to v} &:=  (-1)^{g_v}m_{c\to v}  
\end{align}
We first show that the ordinary BP update is gauge covariant under this transformation. 
\begin{proposition}[Gauge covariance of BP]
    Let $F_{\BP}(m; s; \pi)$ be the BP update as a function of the messages $m$, syndrome $s$, and priors $\pi$. Then, 
    \[
    F_{\BP}(G_g m; s \oplus H_X g; G_g\pi) = G_g F_{\BP}(m; s; \pi)
    \]
\end{proposition}
\begin{proof}
    Consider first the bit to check update, 
    \[
    \bar{m}'_{v\to c} = \pi_v + \sum_{c'\in\cN(v)/c}m_{c'\to v}
    \]
    Each factor is multiplied by $(-1)^{g_v}$, and we thus have $\bar{m}'_{v\to c}\to (-1)^{g_v} \bar{m}'_{v\to c}$. Now, the check to bit equation is, 
    \[
    \bar{m}'_{c\to v} = 2(-1)^{s_c + d_c} \atanh\left(\prod_{w\in\cN(c)/v}\tanh\left[\frac{\bar{m}_{w\to c}}{2}\right]\right)
    \]
    transforming the syndrome,
    \[
    s_c \to s_c \oplus \bigoplus_{u\in\cN(c)} g_u.
    \]
    Moreover, since $\tanh, \atanh$ are odd functions, one accumulates a net factor of 
    \[
    (-1)^{\bigoplus_{u\in\cN(c)} g_u + \bigoplus_{w\in\cN(c)/v} g_w} = (-1)^{g_v}
    \]
    resulting in $\bar{m}'_{c\to v} \to (-1)^{g_v} \bar{m}'_{c\to v}$. As a result, 
\[
    F_{\BP}(G_g m; s \oplus H_X g; G_g\pi) = G_g F_{\BP}(m; s; \pi)
    \]
\end{proof}

\begin{lemma}[Mapping to ordinary BP]\label{lem:ordinaryBPmap}
    For every $i \geq 0$ the discrete inference machine satisfies the following relations to the ordinary BP iterations using the cumulative correction $z$. Immediately before the feedback, 
    \[
    \hat{L}^{(i)}_{v} =  (-1)^{z_v^{(i-1)}} \bar{L}^{(i)}_{v}
    \]
    and after feedback, 
    \[
    m^{(i)} = G_{z^{(i)}} \bar{m}^{(i)}.
    \]
    Moreover, if ordinary BP clears the syndrome at iteration $i$, so does the discrete autonomous machine.
\end{lemma}
\begin{proof}
    The claim follows at $i=0$ as there are no flips at iteration $0$ and both the dynamics are identical. 
    Suppose the assertion is true at $i-1$. The current machine syndrome is, 
    \[
    s^{(i-1)} = s\oplus H_X z^{(i-1)}
    \]
    Applying the covariance identity gives, 
    \[
    \hat{L}_v^{(i)} = (-1)^{z_v^{(i-1)}} \bar{L}_v^{(i)}
    \]
    Upon which, the machine applies the feedback $f_v^{(i)} = \1[ \hat{L}_v^{(i)} > 0]$. As a result, 
    \begin{align*}
        z_v^{(i)} &= z_v^{(i-1)} \oplus f_v^{(i)} \\ 
        &=  z_v^{(i-1)} \oplus \1[ \hat{L}_v^{(i)} > 0] \\ 
        &=  z_v^{(i-1)} \oplus \1[  (-1)^{z_v^{(i-1)}} \bar{L}_v^{(i)} > 0] \\ 
        &= \1[\bar{L}_v^{(i)} > 0]
    \end{align*}
    Consequently, 
    \[
    s^{(i)} = s \oplus H_X\left[\1[\bar{L}_v^{(i)} > 0]\right]
    \]
    Thus, if ordinary BP clears the syndrome at iteration $i$, so does the discrete machine, with the identical correction. 
\end{proof}

\smsection{Optimal threshold on the repetition code}

We now analyze the performance of the local inference process. To do so, we work with the code on a odd-$d$ periodic chain, and define a hierarchy of intervals covering the code. For each $1 \leq i \leq d$ we define, 
\begin{equation}
    r_i := \lfloor \frac{i-1}{2}\rfloor, \qquad k_i := 2r_i + 1, \qquad W_v^{(i)} : = \{v-r_i , \dots, v+r_i\}
\end{equation}
Now, the following lemma shows that a nonzero leftover error on bit $v$ at iteration $(i)$ can only happen if the window $W_v^{(i)}$ had the wrong majority vote to begin with.
\begin{lemma}[Window majority controls residual error]\label{lem:repwindowmaj}
    Write the initial error locations as $E = \{v \mid e_v^{(0)} = 1\}$. Then, 
    \[
    e_v^{(i)}=1 \Leftrightarrow | E \cap W_v^{(i)}| > \frac{k_i}{2}
    \]
\end{lemma}
\begin{proof}
    It is convenient to switch to $\pm 1$ bit magnetization for the following argument. Define the initial magnetization, $\sigma_v := (-1)^{e_v^{(0)}}$ for all $v\in \cV$. Further, for the window $W_v^{(i)}$, define the window magnetization $M_v^{(i)} := \sum_{u\in W_v^{(i)}}\sigma_u$. We have the following relation between the finite time residual errors and majority decoding within a window. 
\[
\bar{L}_v^{(i)} = \pi_v \sigma_v M_v^{(i)} \qquad (-1)^{e_v^{(i)}} = \sgn M_v^{(i)}
\]
To see this, write $w_v = |E \cap W_v^{(i)}|$. The internal checks in $W_v^{(i)}$ have exactly two compatible configurations, $e^{(0)}$ restricted to the path and its bitwise complement. The ratio of these probabilities is, 
\[
\log \frac{q^{w_v}(1-q)^{k_i - w_v}}{q^{k_i - w_v}(1-q)^{w_v}} = (2w_v - k_i)\log{\frac{q}{1-q}} = M_v^{(i)} \pi_v 
\]
Thus, we have $\bar{L}_v^{(i)} = \sigma_v \pi_v M_v^{(i)}$. Now, from \cref{lem:ordinaryBPmap} we have that $\hat{L}^{(i)}_{v} =  (-1)^{z_v^{(i-1)}} \bar{L}^{(i)}_{v}$. Also, recall that $z_v^{(i)} = z_v^{(i-1)}\oplus f_v^{(i)}$. Putting these together, 

\begin{align}
    z_v^{(i)} &= z_v^{(i-1)}\oplus f_v^{(i)}\\
    &= z_v^{(i-1)}\oplus \1[\hat{L}_v^{(i)} > 0]\\
    &= z_v^{(i-1)} \oplus \1[(-1)^{z_v^{(i-1)}} \bar{L}^{(i)}_{v} > 0]  \\ 
    &= \1[\bar{L}^{(i)}_{v} > 0]
\end{align}
Thus, for the residual error we have, 
\begin{align}
    e_v^{(i)} &:= e_v^{(0)}\oplus z_v^{(i)} \\ 
    \implies (-1)^{e_v^{(i)} } &= \sigma_v (-1)^{z_v^{(i)}} \\ 
    &= -\sigma_v \sgn \bar{L}_v^{(i)}\\ 
    &=\sgn M_v^{(i)}
\end{align}
Lastly, we note that $M_v^{(i)} = k_i - 2|E\cap W_v^{(i)}|$ to conclude $e_v^{(i)} \Leftrightarrow | E \cap W_v^{(i)}| > \frac{k_i}{2}$.
\end{proof}

To show a threshold for this decoder, we assume that the initial error $e^{(0)}$ occurs with rate $p$ independently on each qubit. Actually, we only need a weaker property that this distribution satisfies, which is called $p-$boundedness. We now recall the definition of $p-$bounded noise.
\begin{defn}[$p$-bounded noise~\cite{gottesmanFaultTolerantQuantumComputation2014}]
A distribution over $\cV$ is called $p$-bounded if, for every
finite subset \(S\subseteq \cV\),
\begin{equation}
    \Pr\left[S\subseteq E\right]
    \leq
    p^{|S|}.
\end{equation}
\end{defn}
The key property that we will plug into the local majority argument of \cref{lem:repwindowmaj} is that $p-$bounded noise clusters in the following sense: the probability of a window having more than half of its bits being erroneous is bounded exponentially in the window size.

\begin{lemma}[$p-$bounded tail]\label{lem:pboundedtail}
    For every set $S\subseteq \cV$ of $m$ sites, and $E$ drawn from a $p$ bounded distribution over $\cV$, 
    \[
    \Pr(|E\cap A| \geq m/2) \leq e^{-D_p m}
    \]
    where $D_p := -\log\left(2\sqrt{p(1-p)}\right)$ for $p < 1/2$.
\end{lemma}
\begin{proof}
    Denote $X_v:= \1[v\in E]$ and $X_S := \sum_{v\in A} X_v$. For any $\theta > 0$, we have 
    \begin{align}
        \E e^{\theta S} &= \E \prod_{v\in S}e^{\theta X_v} \\ 
        &= \E \prod_{v\in S}\left(1 + (e^\theta - 1)X_v\right) \\ 
        &= \E \sum_{B\subseteq A} (e^\theta -1)^{|B|} \\ 
        &=  \sum_{B\subseteq A} (e^\theta -1)^{|B|}\Pr(B \subseteq E) \\ 
        &\leq \sum_{B\subseteq A} \left[p (e^{\theta}- 1)^{|B|}\right] \\ 
        &\leq (1-p +pe^\theta)^m
    \end{align}

    We now apply Markov's inequality, 
    \begin{align}
        \Pr(S \geq m/2) &= \Pr(e^{\theta S} \geq e^{\theta m/2})\\ 
        &\leq \frac{\E e^{\theta S}}{e^{\theta m/2}} \\ 
        &\leq \left((1-p+pe^\theta) e^{-\theta/2}\right)^m
    \end{align}
    Now we choose $\theta$ such that $e^\theta=(1-p)/p$, and conclude. 
\end{proof}

We are now ready to state and prove the threshold theorem for the discrete local inference machine operating on the repetition code.  
\begin{theorem}[Rep code local inference] \label{thm:repcodesupp}
    The discrete local inference machine operating on the distance$-d$ repetition code achieves a logical error rate, 
    \[
    p_L \leq e^{-D_p d}
    \]
    where $D_p > 0$ is a constant whenever the physical error rate $p < p_{\textrm{th}} = 0.5$. Moreover, it clears the syndrome in $i_{\clr}$ iterations satisfying, 
    \[
    \E[i_{\clr}] = O(\log d).
    \]
\end{theorem}
\begin{proof}
    Now, if $i_{\clr}$ is the clearing time, then the failure event is exactly $\cF := \{e^{(i_{\clr})} \neq 0\}$. By \cref{lem:repwindowmaj}, $\cF$ is exactly the event $|E| > d/2$, thereby, 
    \[
    \Pr(\cF) \leq e^{-D_p d}
    \]
    by an application of \cref{lem:pboundedtail}. As a result $i_{\clr}\leq d$ for any initial error.

    Define the $i-$level failure event in $W_v^{(i)}$,
    \[
    \cA_i := \{\exists v\in \cV \mid |E\cap W_v^{(i)}| > k_i/2\} .  \]
    A union bound gives, 
    \[
    \Pr(\cA_i) \leq \sum_{v\in \cV} \Pr\left(|E\cap W_v^{(i)| > k_i/2}\right) \leq d e^{-D_p k_i}
    \]
     Furthermore, the event $\cF \cup\{i_{\clr} > i\}$ can only happen if $\cA_i$ happens, i.e., $\cF \cup\{i_{\clr} > i\}\subseteq \cA_i$. To see this, note that \cref{lem:repwindowmaj} implies that if $\cA_i$ does not occur, then $e^{(i)}=0$ and decoding succeeds. Since $k_i \geq i-1$, we get, 
     \[
     \Pr(\{i_{\clr} > i\}) \leq d\cdot e^{-D_p (i-1)}
     \]
     Now if we set $I := 1 + \lceil (\log d)/D_p \rceil$, 
     \begin{align*}
        \E[i_{\clr}] &= \sum_{i=0}^{\infty}\Pr(\{i_{\clr} > i\}) \\ 
        &\leq  \sum_{i=0}^{I-1}1 + \sum_{i=I}^{\infty}\Pr(\{i_{\clr} > i\}) \\ 
        &\leq I + \sum_{i=I}^{\infty} d\cdot e^{-D_p (i-1)} \\ 
        &\leq 1 + \lceil \frac{(\log d)}{D_p }\rceil + \frac{1}{1-e^{-D_p}} \\ 
        &\leq O(\log d)
     \end{align*}
\end{proof}

\smsection{Numerical Details}

The continuous dynamics is simulated at a discretization of $dt$, with  Poisson clock rings simulated by drawing a waiting time $\sim \text{Exp}(|S(t)|)$ where $S(t)$ is the set of currently active defects. Each $(d,p)$ data point in the numerical result is a Monte-Carlo estimate of the logical error rate with either $500,000$ shots, or $100$ logical errors, whichever occurs first. Additionally, all the LLRs used in the numerics are clamped at a fixed value of $\texttt{CLAMP}$. The inference dynamics is run for a maximum of $t_{\max}$ time, with a discretization of $dt$ and a relaxation rate of $\tau$. If the syndrome does not converge uptil $t_{\max}$, the sample is declared a logical error. Moreover, the $t_{\clr}$ times are shown only for samples whose syndrome clears at $t \leq t_{\max}$.

\smsubsection{Additional numerics}
We first report some additional data on the 2D toric code in \cref{fig:toric_supp}. In panel (a), we show the logical error rate as a function of code size at $p=0.06$ for both the continous machine and MWPM. As expected, we observe the local algorithm suppresses the logical error with a slower slope $\alpha$ than MWPM, where we have $p_L \sim (p/p_{\textrm{th}})^{\alpha p}$. Scanning across a range of $p$, we next plot the slope $\alpha(p)$ for both the decoders in panel (b). For MWPM, as expected, the slope approaches $0.5$ as $p\to 0$. However, we note that it approaches $\approx 0.25$ for the inference machine. Since one needs the inference velocity $\sim 1/\tau$ to be twice that of the Poisson clock rate~\cite{lakeFastOfflineDecoding2025,rigorousinfmachine}, we expect the threshold to change with $\tau$: saturating to $\sim 8\%$ as $\tau\to 0$ and saturating to zero as $\tau\to 0.5$. To study the robustness for $\tau \lesssim 0.2$, we show a plot of the logical error rates for relaxation times (c) $\tau=0.15$ and (d) $\tau=0.2$ in \cref{fig:toric_supp}. The inset of (c) shows the logical error rate at $p=0.07$ for $\tau \in \{0.1,0.15,0.2\}$.

We further report the $t_{\clr}$ histograms of all the codes studied for the continuous inference machine in \cref{fig:hists_supp}. For the 2D toric code and the point-like 3D toric code, we find that all syndromes clear. For the other two codes, we impose a $t_{\max}=1000$, and count a sample as a logical failure if it does not converge the syndrome within $t < t_{\max}$. The histogram, as well as the expected times in the main text, are shown only for samples that clear the syndrome. 
\begin{table}[h]
\caption{Numerical parameters used for each code family.}
\label{tab:numerical-details}
\begin{ruledtabular}
\begin{tabular}{lccccccc}
                 & $dt$ & $\tau$ & $t_{\max}$ & $p_{\mathrm{th}}$
                 & $q$ & \texttt{CLAMP} & \textsc{Evidence} \\
\colrule
2D TC            & $0.001$ & $0.1$ & --- & $0.08$
                 & $0.05$ & $1000$ & $1/4$ \\
3D TC point      & $0.001$ & $0.1$ & --- & $0.02$
                 & $0.05$ & $1000$ & $1/6$ \\
3D TC membrane   & $0.01$  & $0.1$ & $1000$ & $0.19$
                 & $0.125$    & $1000$ & $1/4$ \\
BB               & $0.01$  & $0.1$ & $1000$ & $0.06$
                 & $0.05$ & $1000$ & $1/6$ \\
\end{tabular}
\end{ruledtabular}
\end{table}

\begin{figure}[t]
    \centering
    \includegraphics[width=1.0\linewidth]{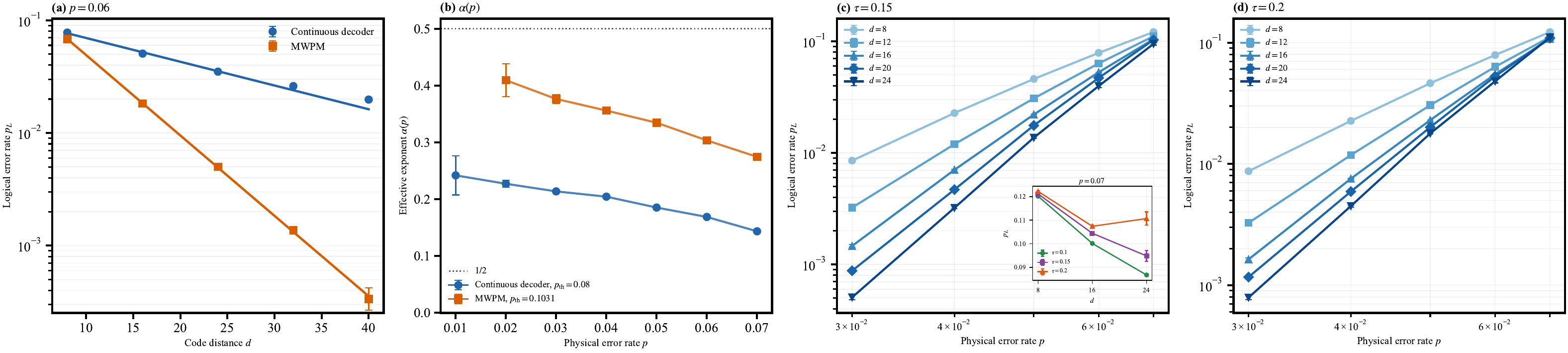}
    \caption{The 2D toric code: (a) Logical error rate as a function of code size of the continuous inference machine and MWPM at $p=0.06$. (b) Extracted slope $\alpha(p)$ from the fit $p_L = A(p) \cdot (p/p_{\textrm{th}})^{\alpha (p) \cdot d}$ for both the decoders. (c-d) Performance of the continuous machine using $\tau \in \{0.15,0.2\}$.}
    \label{fig:toric_supp}
\end{figure}

\begin{figure}[b]
    \centering
    \includegraphics[width=1.0\linewidth]{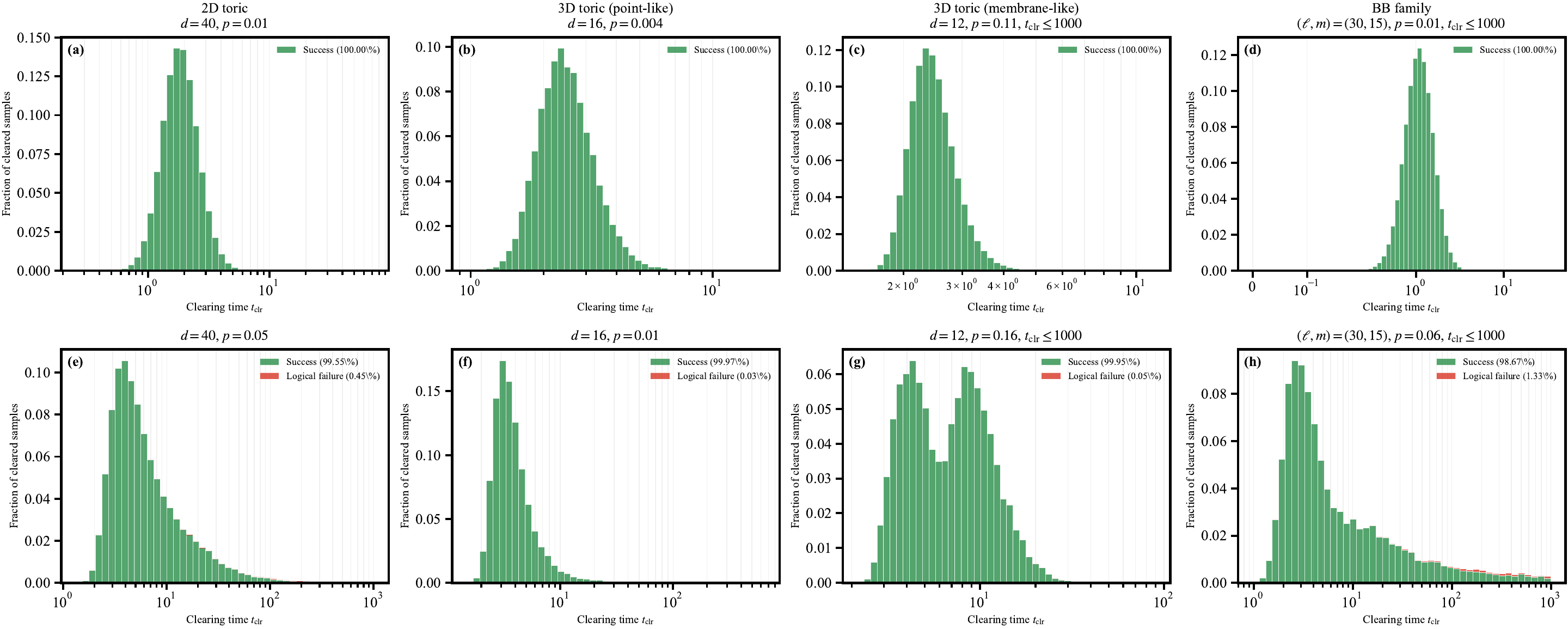}
    \caption{Histogram of the clearing time of the continuous inference machine. (a-d) 2D TC, point-like 3D TC, membrane-like 3D TC, BB family for the smallest noise rate studied, and (e-h) for the largest sub-threshold noise rate. }
    \label{fig:hists_supp}
\end{figure}

\end{document}